\documentclass[12pt,a4paper]{article}

\usepackage{jheppub}
\usepackage{amsmath}
\usepackage{amssymb}
\usepackage{amsthm}
\usepackage{array}
\usepackage{dsfont}
\usepackage{graphicx}
\usepackage{mathtools}
\usepackage{microtype}
\usepackage{multirow}
\usepackage{bigstrut}
\usepackage{setspace}
\usepackage{slashed}

\usepackage[usenames,dvipsnames]{xcolor}
\usepackage{color}
\definecolor{darkblue}{rgb}{0.1,0.1,.7}

\newtheorem{lemma}{Lemma}
\newtheorem{theorem}{Theorem}

\newcommand{\calL}{\mathcal{L}}
\newcommand{\GNY}{{\rm GNY}}

\newcommand{\eps}{\epsilon}

\newcommand{\Spin}{\mathrm{Spin}}
\newcommand{\Sp}{\mathrm{Sp}}
\newcommand{\SL}{\mathrm{SL}}
\newcommand{\SO}{\mathrm{SO}}
\newcommand{\SU}{\mathrm{SU}}
\newcommand{\upO}{\mathrm{O}}
\newcommand{\U}{\mathrm{U}}
\newcommand{\Vrep}{\mathrm{Vec}}

\newcommand{\R}{\mathbb{R}}
\newcommand{\C}{\mathbb{C}}
\newcommand{\HQ}{\mathbb{H}}
\newcommand{\Z}{\mathbb{Z}}

\newcommand{\newHeisenberg}{orthogonal Heisenberg}
\newcommand{\NewHeisenberg}{Orthogonal Heisenberg}

\makeatletter
\def\@fpheader{\ }
\makeatother

\DeclarePairedDelimiter\bra{\langle}{\rvert}
\DeclarePairedDelimiter\ket{\lvert}{\rangle}
\DeclarePairedDelimiter\abs{\lvert}{\rvert}
\DeclareMathOperator{\clos}{clos}

\title{Classifying GNY-like models}
\author[1,2]{Matthew S.~Mitchell}
\author[1]{and David Poland}
\affiliation[1]{Department of Physics, Yale University, 217 Prospect St, New Haven, CT 06520, USA}
\affiliation[2]{Department of Physics, University of Pisa and INFN, Largo Pontecorvo 3, I-56127 Pisa, Italy}
\emailAdd{matthew.mitchell@yale.edu}
\emailAdd{david.poland@yale.edu}

\date{} 
\abstract{
  We perform a systematic classification of (2+1)d Gross--Neveu--Yukawa-like models built out of one or more 4-component Dirac fermions and $M$ scalar fields, which preserve an O($M$) symmetry rotating the scalars. We then identify the perturbative fixed points of these models in the $4-\epsilon$ expansion. Our classification highlights several targets for the conformal bootstrap and reveals a new fixed point with $M=3$, which we call the ``orthogonal Heisenberg'' CFT.
}

\begin{document}

\maketitle

\section{Introduction and motivation}

Critical universality enables us to make specific, widely applicable statements about critical phases of matter via conformal field theory. The goal of this paper is to make some progress toward classifying (2+1)d critical points of \emph{fermionic fields} with \emph{scalar order parameters}. To this end, we shall classify several modifications of the Gross--Neveu--Yukawa model with different symmetry groups, determine which of them possess nontrivial fixed points in the $4-\eps$ expansion, and compute their scaling dimensions perturbatively.

Our work extends similar analyses that have been performed for multi-scalar fixed points, see e.g.~\cite{toledano1985renormalization, Osborn:2017ucf,Codello:2020lta,  Osborn:2020cnf, Hogervorst:2020gtc, Rong:2023xhz}, and is similar in spirit to the recent studies of fermion-scalar fixed points in~\cite{Pannell:2023tzc} and~\cite{Jack:2024sjr}. Compared to these works, we will perform here a more systematic group theory based classification of a specific class of well-motivated models: (2+1)d GNY models containing $N_D$ 4-component Dirac fermions which preserve an $\upO(M)$ symmetry rotating the $M$ scalar fields. 

Our classification includes several universality classes that appear extensively in the literature (specifically the chiral Ising, chiral XY, and chiral Heisenberg models), as well as two that are not well-studied.  The first, which we call the ``quarter GNY'' model, is mentioned in table 4 of~\cite{Pannell:2023tzc} and section 12 of~\cite{Jack:2024sjr}, while the second, which we call the ``\newHeisenberg'' model, does not appear in the literature as far as we are aware. Both of these emerge from an explicit breaking of (3+1)d Lorentz symmetry which preserves the (2+1)d Lorentz subgroup.

To start off, let us consider a GNY model with $N_D$ 4-component Dirac fermions coupled to one scalar, with the Lagrangian 
\begin{equation}\label{eq:gnyDirac}
  \calL = - \frac 1 2 (\partial\phi)^2
  - i \frac 1 2 \bar \Psi_a \slashed\partial \Psi_a
  - \frac 1 2 m^2 \phi^2
  - \frac \lambda 4 \phi^4
  - i \frac g 2 \phi \Psi_a^\dagger \Psi_a.
\end{equation}
For clarity, we shall always use $\Psi$ for 4-component spinors and $\psi$ for 2-component spinors. This theory is particularly interesting because it admits a symmetry enhancement in (2+1)d: There exists a basis in which the first three Dirac-representation $\gamma$ matrices are real and block-diagonal, so when $\gamma^3$ is removed, each Dirac fermion factorizes into four Majorana fermions. This causes the model's $\U(N_D)$ symmetry to be enhanced to $\upO(4N_D)$, resulting in a (2+1)d GNY model with Majorana fermions:
\begin{equation} \label{eq:gnyMajorana} \calL_\GNY = -\frac12 (\partial \phi)^2 - i \frac{1}{2}
  \psi_i \slashed{\partial} \psi_i -\frac{1}{2}m^2\phi^2 -\frac{\lambda}{4}\phi^4 - i \frac{g}{2}
  \phi \psi_i\psi_i.
\end{equation}
This model has been studied pertubatively in quite some detail in both the large-$N$~\cite{Gracey:1992cp,Derkachov:1993uw,Gracey:1993kc, Petkou:1996np, Moshe:2003xn, Iliesiu:2015qra, fei2016yukawa, Manashov:2017rrx, Gracey:2018fwq, Semenoff:2024jqf} and $\epsilon$-expansions~\cite{gracey1990three, rosenstein1993critical, zerf2016superconducting, Gracey:2016mio, fei2016yukawa, Mihaila:2017ble, Zerf:2017zqi, Ihrig:2018hho}. It has also been studied using the conformal bootstrap in \cite{Iliesiu:2015qra, Iliesiu:2017nrv, Erramilli:2022kgp,Mitchell:2024hix}.

Pay close attention to Yukawa coupling in equation (\ref{eq:gnyDirac}), which is written using $\Psi_a^\dagger \Psi_a$ (which would not respect (3+1)d Lorentz symmetry) instead of the more conventional $\bar \Psi_a \Psi_a = \Psi_a^\dagger \gamma^0 \Psi_a$. This is an intentional choice. In the former case, the Yukawa term factorizes as
\begin{equation} \label{eq:yukGNY}
  - i \frac g 2 \phi \Psi_a^\dagger \Psi_a
  =- i \frac g 2 \phi (\psi_a^{Lr} \psi_a^{Lr} + \psi_a^{Li} \psi_a^{Li}
  + \psi_a^{Rr} \psi_a^{Rr} + \psi_a^{Ri} \psi_a^{Ri}),
\end{equation}
while in the latter, it factorizes as
\begin{equation} \label{eq:yukCI}
  - i \frac g 2 \phi \bar\Psi_a \Psi_a
  =- i \frac g 2 \phi (\psi_a^{Lr} \psi_a^{Lr} + \psi_a^{Li} \psi_a^{Li}
  - \psi_a^{Rr} \psi_a^{Rr} - \psi_a^{Ri} \psi_a^{Ri}).
\end{equation}
Here, the upper indices $L$ and $R$ refer to the left- and right-handed components of each Dirac spinor, while $r$ and $i$ refer to the real and imaginary parts. For the mathematical details of this decomposition, see section \ref{sec:spinor3d}, as well as appendix C of \cite{Erramilli:2022kgp}.

These expressions are superficially similar (and as we shall see, have extremely similar conformal fixed points in the $4-\epsilon$ expansion), but their global symmetry groups are quite different. In equation (\ref{eq:yukGNY}), we sum over all possible values of the upper indices just as we sum over lower indices, and obtain an $\upO(4N_D)$ symmetry. This is expected, as the (2+1)d GNY universality classes have $\upO(N)$ symmetry. But in (\ref{eq:yukCI}), reversing the sign of $\phi$ swaps the left- and right-handed terms, so the symmetry group is
\begin{equation}
  \upO(2N_D) \times \upO(2N_D) \rtimes \Z_2.
\end{equation}
This turns out to be the \textbf{chiral Ising} model, which is also well-studied in the condensed matter literature %\cite{Rosenstein:1993zf,Mihaila:2017ble,fei2016yukawa,Gracey:1990wi,Gracey:1992cp,Gracey:1993kc,Derkachov:1993uw,Manashov:2017rrx,Zerf:2017zqi,Ihrig:2018hho} 
\cite{vojta2000quantum, vojta2003quantum,Moon:2012rx,Herbut:2014lfa,Rosenstein:1993zf,Mihaila:2017ble,fei2016yukawa,Janssen:2014gea,Vacca:2015nta,Chandrasekharan:2013aya,Wang:2014cbw,Li:2014aoa,Huffman:2017swn,Hesselmann:2016tvh,Gracey:1990wi,Gracey:1992cp,Vasiliev:1992wr,Vasiliev:1993pi,Gracey:1993kb,Gracey:1993kc,Derkachov:1993uw,Gracey:2017fzu,Manashov:2017rrx,gracey1990three,Gracey:1991vy,Luperini:1991sv,Zerf:2017zqi,Ihrig:2018hho,Huffman:2019efk,Liu:2019xnb} 
(though note that many sources do a poor job of distinguishing between chiral Ising and GNY). The distinction between these models is very important, because phase transitions in the GNY universality class spontaneously break parity, while phase transitions in the chiral Ising universality class preserve a parity symmetry while breaking the $\Z_2$ part of the global symmetry group.

The chiral Ising model is not the only possible modification of the GNY model. If we promote $\phi$ to a two-component flavor vector of scalar fields, we could write down a Lagrangian of the form
\begin{equation}\label{eq:chiXY}
  \calL = \calL_{\mathrm{free}}
  - \frac \lambda 4 \abs{\phi}^4
  - i \frac g 2 \bar\Psi_a \left( \phi_1 + i\phi_2 \gamma_5 \right) \Psi_a.
\end{equation}
This is the so-called \textbf{chiral XY} model (also known as the ``Nambu--Jona--Lasinio--Yukawa model''), which appears as a universality class in Dirac and Weyl semimetals~\cite{Zerf:2017zqi} and is known to display emergent supersymmetry at $N=2$~\cite{fei2016yukawa}. It has also been discussed recently as describing the semimetal-to-Kekul\'e transition on the honeycomb lattice and the twist-tuned transition in moir\'e bilayer graphene~\cite{Hawashin:2025cua}.

In this paper, we perform a classification of all ``GNY-like'' models, which we define to have the following properties:
\begin{itemize}
\item $N_D$ Dirac fermion fields $\Psi$ in the fundamental representation of $\U(N_D)$, which may admit a symmetry enhancement in (2+1)d.
\item Lorentz symmetry in (2+1)d. We permit models to break Lorentz symmetry in (\textbf{3}+1)d, as long as they preserve an $\SO(2,1)$ subgroup.
\item An order parameter $\phi$, which is a vector of $M$ real scalar fields.
\item A Yukawa coupling of the form $\calL_Y = - i \bar\Psi_a \left( \phi_m S_m \right) \Psi_a$, where S is a vector of $M$ spinor matrices, which does \textbf{not} explicitly break the $\upO(M)$ symmetry of $\phi$.
\end{itemize}
The purpose of the first two conditions is to limit our attention to critical points that could arise straightforwardly in condensed matter systems, particularly those with strongly correlated electrons.  The purpose of the third and fourth is to ensure that the phase transition, in the Landau paradigm, corresponds to the spontaneous breaking of an $\upO(M)$ symmetry.

In future work, we hope to use similar methods classify critical points for which $\upO(M)$ is not preserved, in order to see if there are any new universality classes which can be obtained from GNY ones by RG flows (as is the case for the Heisenberg and cubic universality classes \cite{Chester:2020iyt}). We also leave to future work the more general study of fixed points involving (2+1)d Majorana fermions which do not embed into 4-component Dirac fermions. 

\section{Representation theory}

We'll begin by carefully defining our notation and describing our setup. Along the way we'll prove various lemmas needed for our classification. For readers in a hurry, the most important takeaways are the definitions in subsection \ref{sec:globsym}, the expression for the symmetry group in theorem \ref{th:finalG}, the enumeration of models in tables \ref{tab:gnylist1}, \ref{tab:gnylist2}, and \ref{tab:gnylist3}, and the list of fixed points in table \ref{tab:ssb}.

\subsection{Review of spinors in (2+1)d} \label{sec:spinor3d}

In $d$-dimensional Euclidean space, a \textbf{spinor} is a representation of $\Spin(d)$, the unique double-cover of $\SO(d)$. Similarly, a Lorentzian spinor is a representation of $\Spin(d-1,1)$, which double-covers the (proper orthochronous) Lorentz group $\SO(d-1,1)^+$. In either signature, a spinor representation corresponds to a choice of Dirac matrices $\gamma$ which satisfy the Clifford algebra.
\begin{equation}
  \gamma^\mu \gamma^\nu + \gamma^\nu \gamma^\mu = 2 \eta^{\mu\nu}.
\end{equation}
The Spin groups in 3 and 4 dimensions are:
\begin{equation}
  \begin{split}
    \Spin(3) &= \SU(2),\\
    \Spin(2,1) &= \SL(2,\R) = \Sp(2,\R) = \mathrm U (1,1),\\
    \Spin(4) &= \SU(2) \times \SU(2),\\
    \Spin(3,1) &= \SL(2,\C) = \Sp(2,\C).
  \end{split}
\end{equation}

We'd like to write a GNY Lagrangian in terms of \textbf{2-component Majorana spinors}. This is easy in (2+1)d Lorentzian signature, because the Clifford algebra has a real representation given by
\begin{align}
  \gamma^0 = i\sigma^2, &&
  \gamma^1 = \sigma^1, &&
  \gamma^2 = \sigma^3.
\end{align}
In $(4-\epsilon)$d, we can choose a basis $\tilde\gamma^\mu$ that simplifies to the real (2+1)d basis when $\epsilon=1$. In block form, these are
\begin{align}
  \tilde\gamma^0 = \begin{pmatrix} \gamma^0 & 0 \\ 0 & -\gamma^0 \end{pmatrix},
  &&
     \tilde\gamma^1 = \begin{pmatrix} \gamma^1 & 0 \\ 0 & -\gamma^1 \end{pmatrix},
  &&
     \tilde\gamma^2 = \begin{pmatrix} \gamma^2 & 0 \\ 0 & -\gamma^2 \end{pmatrix},
  &&
     \tilde\gamma^3 = \begin{pmatrix} 0 & -i1 \\ i1 &0 \end{pmatrix},
\end{align}
\begin{equation}
  \tilde\gamma^5 = i\tilde\gamma^0 \tilde\gamma^1 \tilde\gamma^2 \tilde\gamma^3 =
  \begin{pmatrix} 0 & 1 \\ 1 &0 \end{pmatrix}.
\end{equation}
The first 3 of these matrices are block-diagonal, so our 4-component complex spinors can be factorized into left- and right-handed 2-component complex spinors when $\tilde\gamma^3$ is removed. We write this as
\begin{equation}
  \Psi =
  \begin{pmatrix}
    \psi^L \\ \psi^R
  \end{pmatrix}.
\end{equation}
Because of the sign in $\tilde\gamma^0$, the conjugate of $\Psi$ is
\begin{equation}
  \bar\Psi =
  \begin{pmatrix}
    \bar\psi^L & -\bar\psi^R
  \end{pmatrix}.
\end{equation}
Furthermore, because $\gamma^\mu$ are real, we can decompose each 2-component spinor into its real and imaginary parts, each of which is a Majorana spinor.

In the convention of \cite{Iliesiu:2015qra}, the product of two Majorana spinors is
\begin{equation} \label{eq:majoranaProduct}
  \chi\psi \equiv (\gamma^0)^{\alpha\beta} \chi_\alpha \psi_\beta.
\end{equation}
The matrix $\gamma^0 = i\sigma^2$ is antisymmetric, and fermions are made of Grassmann numbers, so this product is \textbf{commutative}. Using this convention, we can write the product of two \emph{complex} 2-component spinors as
\begin{equation} \label{eq:complexprod}
  \begin{split}
    \bar\chi \psi &= (\chi_r - i \chi_i) \gamma^0 (\psi_r + i \psi_i)\\
                  &= \chi_r \psi_r + \chi_i \psi_i + i (\chi_r \psi_i - \chi_i \psi_r),
  \end{split}
\end{equation}
where the terms in the last line have an implicit $\gamma^0$ as in (\ref{eq:majoranaProduct}). So when restricted to 2+1 dimensions, a complex 4-component Dirac spinor can be decomposed into 4 real 2-component Majorana spinors.

Remember, our goal is to write down a coupling of the form $\bar\Psi S(\phi) \Psi$ that reduces to a sum of Yukawa couplings in (2+1)d, where $S(\phi)$ has spinor indices and is linear in $\phi$. When expanded out into the Majorana basis, such a coupling looks like
\begin{align}
  \label{eq:ymatcomp}
  \bar\Psi S(\phi) \Psi
  &=\begin{pmatrix} \bar\psi^L & \bar\psi^R\end{pmatrix} 
    \begin{pmatrix} 1&0\\0&-1 \end{pmatrix}
    \begin{pmatrix} S_{11} & S_{12} \\ S_{21} & S_{22} \end{pmatrix}
    \begin{pmatrix} \psi^L \\ \psi^R\end{pmatrix} \\
  \label{eq:ymatreal}
  &=\begin{pmatrix} \psi^L_r & \psi^L_i & \psi^R_r & \psi^R_i\end{pmatrix}
    \begin{pmatrix}
      R_{11} & R_{12} & R_{13} & R_{14} \\
      R_{21} & R_{22} & R_{23} & R_{24} \\
      R_{31} & R_{32} & R_{33} & R_{34} \\
      R_{41} & R_{42} & R_{43} & R_{44}
    \end{pmatrix}
    \begin{pmatrix} \psi^L_r \\ \psi^L_i \\ \psi^R_r \\ \psi^R_i\end{pmatrix}.
\end{align}
where $R_{ij}$ are scalars (in the second line, we have folded the sign from $\tilde\gamma^0$ into the $R_{ij}$ terms).

For our theory to be unitary, the Yukawa coupling $g$ must be real, so equation (\ref{eq:complexprod}) imposes the following conditions on $S$ and $R$:
\begin{itemize}
\item The $R_{ij}$ matrix in (\ref{eq:ymatreal}) must be real and symmetric.
\item If $S_{ij}$ are complex scalars (for instance, if $S(\phi)$ is built from $1$, $\tilde\gamma^3$, and $\tilde\gamma^5$), then the matrix
  \begin{equation}
    \begin{pmatrix} 1&0\\0&-1 \end{pmatrix}
    \begin{pmatrix} S_{11} & S_{12} \\ S_{21} & S_{22} \end{pmatrix}
  \end{equation}
  must be Hermitian. In this case, the corresponding block of the $R$ matrix is $\text{Re}[S_{ij}]\cdot \mathbf 1 + \text{Im}[S_{ij}]\cdot (-i\sigma^2)$.
  
\end{itemize}

\subsection{Global symmetry} \label{sec:globsym}

If we have $N_D$ 4-component Dirac fermions in our GNY-like theory (and therefore $N = 4N_D$ Majorana fermions), the Yukawa coupling can be rewritten in the form
\begin{equation}
  \calL_Y = \frac {ig} 2 \psi_a^A R^{AB}(\phi) \psi_a^B,
\end{equation}
where lower indices correspond to the global $\upO(N_D) < U(N_D)$ symmetry of the Dirac fermions and upper indices correspond to the specific Majorana spinor (i.e.\ the choice of $L,R \times r,i$) in each 4-tuple. If we assume $\phi$ is a real, $M$-dimensional vector and $R$ is linear in $\phi$, this becomes
\begin{equation} \label{eq:yukterm}
  \calL_Y = \frac {ig} 2 \phi_m \psi_a^A R_m^{AB} \psi_a^B.
\end{equation}
$R_m$, in this case, is a vector of $M$ 4x4 real symmetric matrices.

Up to an orthogonal change of basis, \textbf{every GNY-like theory should correspond to a choice of $M$ and a set of $R_m$}. A set of $R$ for several GNY-like models can be seen in table \ref{tab:previouslyKnownModels}.
\begin{table}[h]
  \small
  \centering
  \begin{tabular}[pos]{|c|c|c|c|}
    \hline
    Model & $M$ & Symmetry group $G$ & $\{R_m\}$, up to coupling constant\\
    \hline
    GNY& 1 & $\upO(4N_D)$ & \footnotesize $ \begin{pmatrix} 1 & 0 & 0 & 0 \\ 0 & 1 & 0 & 0 \\
                           0 & 0 & 1 & 0 \\ 0 & 0 & 0 & 1 \end{pmatrix} $\\
    \hline
    Chiral Ising & 1 & $\upO(2N_D)^2 \rtimes \Z_2 $ & \footnotesize $\begin{pmatrix} 1 & 0 & 0 & 0 \\ 0 & 1 & 0 & 0 \\ 0 & 0 & -1 & 0 \\ 0 & 0 & 0 & -1 \end{pmatrix} $\\
    \hline
    Chiral XY & 2 & $\upO(2N_D) \rtimes \upO(2)$ & \footnotesize $\begin{pmatrix} 1 & 0 & 0 & 0 \\ 0 & 1 & 0 & 0 \\ 0 & 0 & -1 & 0 \\ 0 & 0 & 0 & -1 \end{pmatrix}, \begin{pmatrix} 0 & 0 & 0 & 1 \\ 0 & 0 & -1 & 0 \\0 & -1 & 0 & 0 \\ 1 & 0 & 0 & 0 \end{pmatrix} $\\
    \hline
    Chiral Heisenberg & 3 & $\U(N_D) \times \upO(3)$ & \footnotesize
     $\begin{pmatrix} 1 & 0 & 0 & 0 \\ 0 & 1 & 0 & 0 \\ 0 & 0 & -1 & 0 \\ 0 & 0 & 0 &-1 \end{pmatrix}$,
     $\begin{pmatrix} 0 & 0 & 0 & 1 \\ 0 & 0 & -1 & 0\\ 0 & -1 & 0 & 0 \\ 1 & 0 & 0 & 0 \end{pmatrix}$,
     $\begin{pmatrix} 0 & 0 & 1 & 0 \\ 0 & 0 & 0 & 1\\ 1 & 0 & 0 & 0 \\ 0 & 1 & 0 & 0 \end{pmatrix}$\\
    \hline
  \end{tabular}
  \caption{List of previously studied GNY-like models, along with their $R$-matrices. Note that this table is incomplete and exists to illustrate the conventions of this section; for a full classification, see tables \ref{tab:gnylist1}--\ref{tab:gnylist3}.}
  \label{tab:previouslyKnownModels}
\end{table}

For a given $M$, different matrices correspond to different global symmetry groups $G$. We devote the remainder of this section to determining $G$ in terms of $M$ and $R_m$.

\subsection{Orthogonal subgroups}

\begin{theorem}\label{th:gsubgroup}
  $G$ is the subgroup of $\upO(4N_D) \times \upO(M)$ that preserves $1_{N_D \times N_D} \otimes R_m$.
\end{theorem}
\begin{proof}
  If we examine the kinetic and quartic terms of equation (\ref{eq:gnyMajorana}), we see that $\psi_i \slashed \partial \psi_i$ and $|\phi|$ must be invariant under the action of $G$. These quantities (ignoring the derivative, which commutes with a change of basis on $\psi$) are the norms of real vectors of length $4N_D$ and $M$, so $G \le \upO(4N_D) \times \upO(M)$.

  The Yukawa term must also be invariant. The action of $Q \in \upO(4N_D)$ and $\Omega \in \upO(M)$ on (\ref{eq:yukterm}) is
  \begin{equation}
    \begin{split}
      \calL_Y'
      &= \frac {ig} 2 (\Omega_{mn} \phi_n) (Q^{AC}_{ac}\psi_c^C) R_m^{AB} (Q^{BD}_{bd} \psi_d^D)\\
      &= \frac {ig} 2 \phi_n \psi_c^C (Q^T (1 \otimes \Omega_{mn} R_m) Q)_{cd}^{CD} \psi_d^D,
    \end{split}
  \end{equation}
  which must equal the original $\calL_Y$.
\end{proof}

If we suppress the flavor indices and shuffle things around, we can write this more succinctly as
\begin{equation}\label{eq:gdef}
  \begin{split}
    G &= \left\{ (Q,\Omega) \in \upO(4N_D) \times \upO(M) \mid
        1_{N_D} \otimes R_m = Q^T ( 1_{N_D} \otimes \Omega_{mn} R_n) Q \right\}\\
      &= \left\{ (Q,\Omega) \in \upO(4N_D) \times \upO(M) \mid
        Q \hat R_m Q^T = \Omega_{mn} \hat R_n \right\},
  \end{split}
\end{equation}
where $\hat R_m \equiv 1 \otimes R_m$.

Keep in mind that there are some choices of $R$ for which not all elements of $\upO(M)$ correspond to elements of $G$. For instance, if $M=2$ and
\begin{align}
  R_1 = \begin{pmatrix} 1 & 0 & 0 & 0 \\ 0 & 1 & 0 & 0 \\
    0 & 0 & 1 & 0 \\ 0 & 0 & 0 & 1 \end{pmatrix},
  && R_2 = \begin{pmatrix} 2 & 0 & 0 & 0 \\ 0 & 2 & 0 & 0 \\
    0 & 0 & 2 & 0 \\ 0 & 0 & 0 & 2 \end{pmatrix},
\end{align}
the $\upO(M)$ subgroup is broken because there is \textbf{no} orthogonal $Q \in \upO(4N_D)$ that compensates for a rotation of $\phi_m$. Recall that, in this paper, we are  \textbf{we are restricting our attention to models in which the $\upO(M)$ subgroup (or at least its connected $\SO(M)$ part) is not explicitly broken}. This imposes some constraints on $R$ and $Q$:

\begin{lemma} \label{lem:samespect}
  If $\SO(M)$ is unbroken, then all $R_m$ (and all normalized linear combinations of the form $\phi_m R_m$) must have the same spectrum.
\end{lemma}
\begin{proof}
  If the $\SO(M)$ subgroup of $G$ is unbroken, then for every $\Omega \in \SO(M)$, there is some $Q \in \upO(4N_D)$ such that $(Q,\Omega) \in G$. By equation (\ref{eq:gdef}), $\Omega_{mn} \hat R_n = Q \hat R_m Q^T$, so $\hat R_m$ and $\Omega_{mn} \hat R_n$ are orthogonally equivalent.
\end{proof}

\begin{lemma} \label{lem:qrep}
  If $\SO(M)$ is unbroken, then every inequivalent set of $R_m$ corresponds to a choice of a real 4d representation of $\Spin(M)$, which we call $\rho_\Psi$, and a choice of at most
  \begin{equation}
    \sum_{\substack{i < j\\\Vrep\; \in\; \rho_i \otimes \rho_j }} m_i m_j\quad +
    \sum_{\substack{i \\\Vrep\; \in\; \mathrm{Sym}^2 \rho_i }} m_i
  \end{equation}
  distinct couplings. Here, $m_i$ is the multiplicity of each inequivalent irrep (over $\R$) occurring in $\rho_\Psi$, and $\Vrep$ is the vector representation, and $\mathrm{Sym}^2$ denotes the symmetric square of a representation.

  In other words: for every copy of $\Vrep$ occurring in the (symmetric) product of two sub-irreps of $\rho_\Psi$, there is a coupling that can be tuned without breaking $\SO(M)$. But when we obtain an $m \times m$ square matrix of couplings from $m$ copies of the \emph{same} irrep, it can be diagonalized.
\end{lemma}

\begin{proof}
  We can see by inspection that $Q$ furnishes a 4d (projective) representation of $\SO(M)$ over $\R$. Let's call this representation $\rho_\Psi$. Because we only care about the action on $R_m$, on which $Q=-1$ acts trivially, $\rho_\Psi$ is allowed to be a projective representation in the double cover $\Spin(M)$.

  Assume that $\rho_\Psi$ is block-diagonal in its constituent irreps, and can be written as
  \begin{equation}
    \rho_\Psi = \rho_1 \oplus \rho_2 \oplus \dots
  \end{equation}
  (there is always a basis in which this is true).

  $R_m$ itself must transform in the fundamental (i.e.~vector) rep of $\SO(M)$. This is equivalent to saying that $R_m$ is a rank-1 tensor operator on a real Hilbert space with $\Spin(M)$ symmetry. We can apply Wigner--Eckart to find its matrix elements:
  \begin{equation} \label{eq:wet}
    \bra{\rho s} R^{(\Vrep)}_m \ket{\rho' s'} = \bra{\rho s; (\Vrep) m}\ket{\rho' s'} \langle\rho \|R\| \rho'\rangle .
  \end{equation}
The matrix elements here are indexed by the irreps $\rho$ and their individual state $s$. Na\"ively, it seems we have $n^2$ reduced matrix elements (i.e.~couplings) to tune independently, where $n$ is the number of (possibly identical) irreps in $\rho_\Psi$. This number can be reduced further, because:
\begin{itemize}
\item $R$ must be symmetric. So only $n(n+1)/2$ of the reduced matrix elements are \emph{actually} independent.
\item If $\rho$ does not appear in $\rho' \otimes \Vrep$, the CG coefficients vanish, so we can ignore $\langle\rho' \|R\| \rho\rangle$. Because representations over $\R$ are self-conjugate, this condition is equivalent to $\Vrep \notin \rho \otimes \rho'$.
\item If $\rho_\Psi$ contains $m$ copies of the irrep $\rho$, and $\rho \in \Vrep \otimes \rho$, the corresponding block of equation (\ref{eq:wet}) is a Kronecker product. So we can diagonalize the $\langle\rho' \|R\| \rho\rangle$ part, leaving us with only $m$ distinct couplings.
\item If $\rho \otimes \rho$ contains the vector representation, but the Clebsch--Gordan term is entirely antisymmetric, then $\langle\rho' \|R\| \rho\rangle$ must vanish to make $R$ symmetric. This is readily apparent in $\Spin(2)$, where $3 \otimes 3 = 1 \oplus 3 \oplus 5$, but the $3$ corresponds to the antisymmetric part.

  Therefore, for the coupling to be nonzero, it is necessary that $\Vrep$ appear in the \emph{symmetric square} of $\rho$.
\end{itemize}
  By Schur's lemma, no two representations can correspond to the same set of $R_m$, or to two orthogonally equivalent $R_m$s, and a change of basis on $R_m$ does not alter $\rho_\Psi$. So given a choice of  $\rho_\Psi$, all choices of $R_m$ are equivalent \textbf{unless} they have distinct spectra. These spectra are determined by the reduced matrix elements.
\end{proof}

Note that, if we set some of these couplings equal to each other, we get a symmetry enhancement. This is why the GNY and chiral Ising have different symmetry groups, even though both of them have $\rho_\Psi = 1 \oplus 1 \oplus 1 \oplus 1$.

For an enumeration of possible $\rho_\Psi$ for different $M$, see appendix \ref{app:reps}.

\begin{lemma} \label{lem:disconnected}
  The disconnected part of the $\upO(M)$ subgroup is preserved iff either $M>1$, or $M=1$ and the spectrum of $R$ is symmetric about 0.
\end{lemma}
\begin{proof}
  First, the $M=1$ case:  $\phi$ transforms in the odd rep of $\upO(1) = \Z_2$, so the reflection element takes $R_1$ to $-R_1$. If the spectrum is not symmetric about 0, there is no orthogonal change of basis that can undo this. If the spectrum \emph{is} symmetric about 0, we can find an orthogonal change of basis that reorders the eigenvalues in $R_1$, swapping each with its negation.

  When $M>1$, $\SO(M)$ is a Lie group, and asking whether the disconnected part is preserved is equivalent to asking whether the inner automorphism group of the corresponding Lie algebra has a disconnected part. For all $M>1$, this is true.

 (One can also see by inspection that, for every row in tables \ref{tab:gnylist2} and \ref{tab:gnylist3}, there is a matrix $X$ such that $XR_mX^T = -R_m$ for exactly one $m$ and $XR_mX^T = +R_m$ for all others. All $X$ happen to be diagonal with entries $\pm 1$.)
 \end{proof}

\begin{theorem}
  Given matrices $R_m$ constructed from a spectrum and a choice of $\rho_G$  (see Lemma \ref{lem:qrep}), the symmetry group of a GNY-like model is
  \begin{equation}
    G =
    \begin{cases}
      H &\text{if $M=1$ and the spectrum of R is asymmetric}\\
      H \rtimes \upO(M)&\text{otherwise}
    \end{cases}
  \end{equation}
  or a double cover thereof, where
  \begin{equation}
    H = \left\{ Q \in \upO(4N_D) \mid \forall m, [1_{N_D\times N_D} \otimes R_m,Q] = 0\right\}.
  \end{equation}
\end{theorem}
\begin{proof}
  We introduce the real Lie algebra
  \begin{equation}\label{eq:galg}
    \mathfrak g = \left\{ (X,\omega) \in \mathfrak{so}(4N_D) \oplus \mathfrak{so}(M)
      \mid \omega_{mn} \hat R_n = [X, \hat R_m] \right\}.
  \end{equation}
  which corresponds to the connected part of $G$ via the Lie map (note the similarity to \ref{eq:gdef}). For any element of $\mathfrak g$, we can see that $X$ uniquely determines $\omega$, and $\omega$ determines $X$ up to some matrix $Y$ that commutes with all $R_m$. By the Jacobi identity, the $Y$s with this property form a \textbf{commutant subalgebra}. Now, for each $\omega$, let's choose a canonical $X_0(\omega)$ and redefine $X = X_0(\omega) + Y$.

  Returning to the Lie group, we can see that for any $(X_0(\omega) + Y, \omega) \in \mathfrak g$, there is a corresponding $(Q,\Omega) \in G$ such that
  \begin{align}
    \Omega &= e^\omega \label{eq:lieOmega}, \\
    Q &=  e^{X_0(\omega) + Y} = e^{X_0(\omega)} Q'(\omega, Y), \label{eq:lieQ}
  \end{align}
  where
  \begin{equation}
    Q'(\omega, Y) = \exp(Y + \frac 1 2 [-X_0(\omega), X_0(\omega)+Y]
    + \frac 1 {12} [-X_0(\omega), [-X_0(\omega), \dots]] + \dots)
  \end{equation}
  By equation (\ref{eq:galg}), the set $\{\hat R_m\}$ is mapped to itself under commutation with any $X \in \mathfrak g$, and $Y$ must commute with all $\hat R_m$, so
  \begin{equation}
    \begin{split}
      0 &= [X, [Y, R]] + [Y, [R, X]] + [R, [X, Y]]\\
        &= 0 + [Y, R'] + [R, [X, Y]]\\
        &= 0 + 0 + [R, [X, Y]].
    \end{split}
  \end{equation}
  This means that commutation with any $X$ acts as an automorphism on the commutant subalgebra, so we can rewrite (\ref{eq:lieQ}) as
  \begin{equation}
    Q = e^{X_0(\omega)} e^{Y'(Y, \omega)},
  \end{equation}
  where $Y'$ commutes with all $\hat R_m$ and $Y'=Y$ to leading order in $\omega$.

  It is therefore possible to write any element of the \emph{connected} part of $G$ as $g = (e^{X_0(\omega)},e^\omega)\cdot (e^{Y'},1)$. The $\omega = 0$ subgroup clearly is normal, so the connected subgroup $G$ is a semidirect product of $\SO(M)$ with the commutant subgroup $H = \{e^{Y'}\}$.

  The same argument holds for elements in the disconnected part (which, by Lemma \ref{lem:disconnected}, exists unless $M=1$ and the spectrum of $R$ is asymmetric).  We can write $g$ as $g = (e^{X'_0(\omega)},se^\omega)\cdot (e^{Y'},1)$, where $s$ is the reflection element. So $G = H \rtimes \upO(M) $.
\end{proof}

\subsection{Commutants}

The task of classifying GNY-like models has now been reduced to determining the commutant subgroups $H \le \upO(4N_D)$ for valid sets of $\{R_m\}$. To accomplish this, we will need to prove some basic properties of $H$:

\begin{lemma}\label{lem:commalg}
  The connected part of $H$ is generated by
  \begin{equation}\label{eq:commalg}
    \mathfrak h = \left\{ X \in \mathfrak{so}(4N_D) \mid
      \forall m, [1_{N_D\times N_D} \otimes R_m,X] = 0\right\},
  \end{equation}
  which is a Lie algebra.
\end{lemma}
\begin{proof}
  $\mathfrak h$, as defined above, is clearly closed under linear combinations of $X$. For any two $X_1,X_2 \in \mathfrak h$,
  \begin{equation}
    [R,[X_1,X_2]] = -[X_1,[X_2,R]] -[X_2,[R,X_1]] = 0
  \end{equation}
  by the Jacobi identity, so  $\mathfrak h$ is also closed under commutators. It therefore satisfies the axioms of a Lie algebra.

   $e^{\alpha X}$ commutes with a given matrix for all $\alpha$ if and only if $X$ does, so the connected part of $H$ is generated by $\mathfrak h$.
\end{proof}

\begin{lemma}
  Let the \textbf{closure algebra} $\clos \{R_m\}$ of a set $\{R_m\}$ be the real $C^*$-algebra generated by all $R_m$ and their pairwise products. $\{R_m\}$ and  $\clos\{R_m\}$ have the same commutant algebra.\footnote{This is also true for the closure under commutators or anticommutators.}
\end{lemma}
\begin{proof}
  $\{R_m\} \subseteq \clos\{R_m\}$, so the commutant algebra of $\clos \{R_m\}$ must be contained within the commutant algebra of $\{R_m\}$. Any $X$ that commutes with all $\{R_m\}$ must commute with their products as well, so the commutant algebra of $\clos \{R_m\}$ is no smaller than that of $\{R_m\}$. They are therefore equal.
\end{proof}

$R_m$ are $4\times 4$ matrices, so the closure algebra $\clos\{R_m\}$ must be a 4d representation over $\R$ of some algebra $\mathfrak r$. By Weyl's theorem, it must be completely reducible, so we can choose a basis in which $\clos\{R_m\}$ is a block-diagonal sum of modules (i.e.~irreps):
\begin{equation}
  \clos\{R_m\}(x) \simeq \rho_1(x) \oplus \rho_2(x) \oplus \dots =
  \begin{pmatrix}
    \rho_1(x) & & \\ & \rho_2(x) & \\ & & \ddots
  \end{pmatrix}.
\end{equation}
If we group identical irreps together, this becomes
\begin{equation}
  \clos\{R_m\}(x) \simeq \rho_1(x)^{\oplus m_1} \oplus \rho_2(x)^{\oplus m_2} \oplus \dots,
\end{equation}
where $m_j$ is the multiplicity of $\rho_j$ and $\rho^{\oplus m}$ denotes the direct sum of $m$ copies of $\rho$. \textbf{The same applies to the closure algebra of $\hat R_m = 1_{N_D \times N_D} \otimes R_m$, except the multiplicity of each irrep is increased by a factor of $N_D$.} It turns out that this decomposition allows us to compute $\mathfrak h$ directly, via:

\begin{lemma}
  $\clos\{R_m\}$ can be decomposed into irreducible modules, which can be classified by their endomorphism algebras as either ``real,'' ``complex,'' or ``quaternionic.''\footnote{To be clear, all these irreps are still over $\R$. But $\R$ is not algebraically closed, so the endomorphism algebra of an irrep is isomorphic to a division algebra over $\R$, which is either $\R$, $\C$, or $\HQ$.} If we write this decomposition as
  \begin{equation} \label{eq:irrepDiag}
    \clos\{R_m\} = \bigoplus_{j} \left( \rho^\R_j  \right)^{\oplus m^\R_j}
    \ \oplus\      \bigoplus_{j} \left( \rho^\C_j  \right)^{\oplus m^\C_j}
    \ \oplus\      \bigoplus_{j} \left( \rho^\HQ_j \right)^{\oplus m^\HQ_j},
  \end{equation}
  where $m^{\R,\C,\HQ}_j$ is the multiplicity of each module, then the full commutant subalgebra is
  \begin{equation} \label{eq:hdecomp}
    \mathfrak{h} = \bigoplus_{j} \mathfrak{so}\left( m^\R_j N_D \right)
    \ \oplus\      \bigoplus_{j} \mathfrak{u} \left( m^\C_j  N_D \right)
    \ \oplus\      \bigoplus_{j} \mathfrak{sp}\left( m^\HQ_j N_D \right).
  \end{equation}
\end{lemma}
\begin{proof}
  Suppose we have some matrix $X$ that commutes with all elements of $\clos\{R_m\}$. We are free to choose a basis such that $\clos\{R_m\}$ takes the block-diagonal form shown in (\ref{eq:irrepDiag}). In this basis, we can write $X$ in block form as
  \begin{equation}
    X =
    \begin{pmatrix}
      X_{11} & X_{12} & X_{13} & \dots \\
      X_{21} & X_{22} & X_{23} &  \\
      X_{31} & X_{32} & X_{33} &  \\
      \vdots &        &        & \ddots
    \end{pmatrix},
  \end{equation}
  where $X_{ij}$ is a $\dim \rho_i \times \dim \rho_j$ matrix. Note that $i$ and $j$ \emph{may} correspond to identical representations. The relation $[X,R_m]=0$ is equivalent to
  \begin{equation}
    X_{ij} \rho_j = \rho_i X_{ij} \text{ (with uncontracted indices)}
  \end{equation}
  for all $i,j$. If  $\rho_i \ne \rho_j$, then by Schur's lemma, $X_{ij} = 0$. It follows that \textbf{$X$ must be block-diagonal, where the blocks correspond to sets of identical irreps}.

  If $\rho_i = \rho_j$, then by Schur's lemma, the possible choices for $X_{ij}$ must furnish a division algebra over $\R$. The Frobenius theorem guarantees that any such algebra is isomorphic to $\R$, $\C$, or $\HQ$, so the individual $X_{ij}$ are real matrices satisfying the same multiplication laws as elements of $\R$, $\C$, or $\HQ$. This means that \textbf{each block of $X$ is isomorphic to an anti-Hermitian matrix over the corresponding field} (anti-Hermiticity is necessary for the real embedding to be skew-symmetric, which must be satisfied for elements of $\mathfrak h$).

  Anti-Hermitian matrices generate the classical automorphism groups of the bilinear form with signature $(+,+,\dots,+)$. Over $\R$, $\C$, and $\HQ$, these groups are $\SO(N)$, $\U(N)$, and $\Sp(N)$ respectively. In this case, $N$ is the multiplicity of the corresponding irrep in the decomposition of $\clos \{R_m\}$ (that is, $N=m^{\mathbb F}_j$). However, because our Lagrangian has $N_D$ copies of $R$, the multiplicity is enhanced by a factor of 4, and we obtain the decomposition in (\ref{eq:hdecomp}).
\end{proof}
It follows directly that:
\begin{theorem}\label{th:finalG}
  For a valid set of $\{R_m\}$, the symmetry group of a GNY-like model is
  \begin{equation}\label{eq:finalg}
    \boxed{\begin{split}
      G &= \left( \prod_j \upO\left( N_D m^\R_j  \right)
      \times   \prod_j \U  \left( N_D m^\C_j  \right)
      \times   \prod_j \Sp \left( N_D m^\HQ_j \right)
      \right)\\ &\qquad \rtimes
      \begin{cases}
        \upO(M) &\text{if $M>1$ or $M=1$ and $R$ has an even spectrum}\\
        1 &\text{if $M=1$ and $R$ does not have an even spectrum}
      \end{cases}
    \end{split}}
  \end{equation}
  where $m^{\R,\C,\HQ}_j$ are the multiplicities of each real, complex, and quaternionic module in the closure of $\{R_m\}$.
\end{theorem}

Note that, because $R_m$ are 4d, the multiplicities must sum to at most 4, and if $M=1$, they must sum to exactly 4 (because all irreps are real and 1d). In addition, recall that $R_m$ can be obtained from choosing an initial $R_1$ and a 4d representation of $\SO(M)$.

\subsection{List of GNY-like models} \label{sec:listofmodels}

To recap, GNY models with unbroken $\SO(M)$ symmetry can be built from the following recipe:
\begin{enumerate}
\item Choose the number of scalar fields, $M$. This must be at most 3 (see below).
\item Choose a 4d representation of $\Spin(M)$ over $\R$, under which the Majorana spinors transform. Call this $\rho_\Psi$.
\item Use the Wigner--Eckart theorem to determine selection rules for the reduced matrix elements of $\bra{\rho_\Psi} R^{(\Vrep)} \ket{\rho_\Psi}$, and choose values for the nonvanishing ones. These are the coupling constants.
\item Optionally, set some of the coupling constants to be equal to each other (which may provide an additional symmetry enhancement).
\item Construct the Yukawa coupling matrices $\{R_m\}$ explicitly from the couplings and Clebsch--Gordan coefficients.
\item The theory's symmetry group can now be found by determining the simple modules of all $\{R_m\}$, using equation (\ref{eq:finalg}).
\end{enumerate}

For $M \ge 4$, there are no nontrivial 4d representations of $\SO(M)$ over $\R$, so all models reduce to the GNY model.

A full list of GNY-like models is shown in tables \ref{tab:gnylist1}--\ref{tab:gnylist3}, and the ones corresponding to perturbative critical points are collected in table \ref{tab:ssb}.

\newcommand{\modelrow}[8][-0.8em]{\multirow[c]{2}[2]{8em}[#1]{\centering #2} & \multicolumn{2}{c|}{\bigstrut #3} & #4 & #5 \\ \cline{2-5} & \bigstrut #6 & \multicolumn{2}{c|}{#7} & #8 }
\newenvironment{modeltab}{\noindent
 \begin{tabular}{|w{c}{7.7em}|w{c}{4em}|w{c}{3.2em}|w{c}{4.1em}|w{c}{14.6em}|}
   \hline \modelrow[0em] {Name} {Symmetry}
   { \parbox{3em}{\vspace{0.2\baselineskip}\footnotesize\centering \setstretch{0.0}\# of\\ couplings \vspace{0.2\baselineskip}} }
   {Yukawa term} {$\Psi$ rep} {$R$-modules} {Coupling matrices $R_m$} \\ }{\\ \hline \end{tabular}}
\newcommand{\ncp}{N/A. o critical point}

\begin{table}
  \small \centering
  \begin{modeltab}
    \hline \hline \modelrow
    {GNY} {$\upO(4N_D)$} {1}
    {$\begin{matrix}
      -\frac{i g} {2} \phi \Psi_a^\dagger \Psi_a = \\
      \frac{-i g} {2} \phi (\psi_a^1 \psi_a^1 + \psi_a^2 \psi_a^2 + \psi_a^3 \psi_a^3 + \psi_a^4 \psi_a^4)
    \end{matrix}$}
    {$ 1^{\otimes 4}$} {$1 \oplus 1 \oplus 1 \oplus 1$}
    {\scriptsize $\begin{pmatrix} g & 0 & 0 & 0 \\ 0 & g & 0 & 0 \\ 0 & 0 & g & 0 \\ 0 & 0 & 0 & g \end{pmatrix}$ }
    \\
    \hline \hline \modelrow
    {``Quarter GNY''} {$\upO(3N_D) \times \upO(N_D)$} {2}
    {$\begin{matrix} -\frac{i} {2} g_1 \phi (\psi_a^1 \psi_a^1 + \psi_a^2 \psi_a^2 + \psi_a^3 \psi_a^3) \\ -\frac{i} {2} g_2 \phi \psi_a^4 \psi_a^4\end{matrix}$}
    {$ 1^{\otimes 4}$} {$1_1 \oplus 1_1 \oplus 1_1 \oplus 1_2$}
    {\scriptsize $\begin{pmatrix} g_1 & 0 & 0 & 0 \\ 0 & g_1 & 0 & 0 \\ 0 & 0 & g_1 & 0 \\ 0 & 0 & 0 & g_2 \end{pmatrix}$}
    \\
    \hline \hline \modelrow
    {} {$\upO(2N_D)^2$} {2}
    {$\begin{matrix} -\frac{i} {2} g_1 \phi (\psi_a^1 \psi_a^1 + \psi_a^2 \psi_a^2) \\ -\frac{i} {2} g_2 \phi (\psi_a^3 \psi_a^3 + \psi_a^4 \psi_a^4)\end{matrix}$}
    {$ 1^{\otimes 4}$} {$1_1 \oplus 1_1 \oplus 1_2 \oplus 1_2$}
    {\scriptsize $\begin{pmatrix} g_1 & 0 & 0 & 0 \\ 0 & g_1 & 0 & 0 \\ 0 & 0 & g_2 & 0 \\ 0 & 0 & 0 & g_2 \end{pmatrix}$}
    \\
    \hline \hline \modelrow
    {} {$\upO(2N_D) \times \upO(N_D)^2$} {3} 
    {$\begin{matrix} -\frac{i} {2} g_1 \phi (\psi_a^1 \psi_a^1 + \psi_a^2 \psi_a^2) \\ -\frac{i} {2} g_2 \phi \psi_a^3 \psi_a^3 -\frac{i} {2} g_3 \phi \psi_a^4 \psi_a^4\end{matrix}$}
    {$ 1^{\otimes 4}$} {$1_1 \oplus 1_1 \oplus 1_2 \oplus 1_3$}
    {\scriptsize $\begin{pmatrix} g_1 & 0 & 0 & 0 \\ 0 & g_1 & 0 & 0 \\ 0 & 0 & g_2 & 0 \\ 0 & 0 & 0 & g_3 \end{pmatrix}$}
    \\
    \hline \hline \modelrow
    {} {$\upO(N_D)^4$} {4}
    {$\begin{matrix} -\frac{i} {2} g_1 \phi \psi_a^1 \psi_a^1  -\frac{i} {2} g_2 \psi_a^2 \psi_a^2 \\-\frac{i} {2} g_3 \phi \psi_a^3 \psi_a^3 -\frac{i} {2} g_4 \phi \psi_a^4 \psi_a^4\end{matrix}$}
    {$ 1^{\otimes 4}$} {$1_1 \oplus 1_2 \oplus 1_3 \oplus 1_4$}
    {\scriptsize $\begin{pmatrix} g_1 & 0 & 0 & 0 \\ 0 & g_2 & 0 & 0 \\ 0 & 0 & g_3 & 0 \\ 0 & 0 & 0 & g_4 \end{pmatrix}$}
    \\
    \hline \hline \modelrow
    {Chiral Ising} {$\upO(2N_D)^2 \rtimes \Z_2 $} {1}
    {$\begin{matrix}
      -\frac{i g} {2} \phi \bar \Psi_a \Psi_a = \\
      \frac{-i g} {2} \phi (\psi_a^1 \psi_a^1 + \psi_a^2 \psi_a^2 - \psi_a^3 \psi_a^3 - \psi_a^4 \psi_a^4)
    \end{matrix}$}
    {$ 1^{\otimes 4}$} {$1_1 \oplus 1_1 \oplus 1_2 \oplus 1_2$ }
    {\scriptsize $\begin{pmatrix} g & 0 & 0 & 0 \\ 0 & g & 0 & 0 \\ 0 & 0 & -g & 0 \\ 0 & 0 & 0 & -g \end{pmatrix}$}
    \\
    \hline \hline \modelrow
    {} {$\upO(N_D)^4 \rtimes \Z_2 $} {2}
    {$\begin{matrix} -\frac{i} {2} g_1 \phi (\psi_a^1 \psi_a^1 - \psi_a^2 \psi_a^2) \\ -\frac{i} {2} g_2 \phi (\psi_a^3 \psi_a^3 - \psi_a^4 \psi_a^4)\end{matrix}$}
    {$1^{\otimes 4}$} {$1_1 \oplus 1_2 \oplus 1_3 \oplus 1_4$}
    {\scriptsize $\begin{pmatrix} g_1 & 0 & 0 & 0 \\ 0 & -g_1 & 0 & 0 \\ 0 & 0 & g_2 & 0 \\ 0 & 0 & 0 & -g_2 \end{pmatrix}$}
  \end{modeltab}
  \caption{An exhaustive list of $M=1$ GNY-like models and their symmetry groups. All named models have a perturbative fixed point in the $4-\epsilon$ expansion; the unnamed ones do not.}
  \label{tab:gnylist1}
\end{table}

\begin{table}
  \small \centering
  \begin{modeltab}
    \hline \hline \modelrow
    {Chiral XY} {$\upO(2N_D) \rtimes \upO(2) $} {1}
    {$- \frac i 2 g \bar\Psi_a \left( \phi_1 + i \tilde \gamma_5 \phi_2 \right) \Psi_a$}
    {$2_q \oplus 2_{q+1}$ } {$2^\R \oplus 2^\R$}
    {\scriptsize \setlength\arraycolsep{1pt} $\begin{pmatrix} g & 0 & 0 & 0 \\ 0 & g & 0 & 0 \\ 0 & 0 & -g & 0 \\ 0 & 0 & 0 & -g \end{pmatrix},
    \begin{pmatrix} 0 & 0 & 0 & g \\ 0 & 0 & -g & 0\\ 0 & -g & 0 & 0 \\ g & 0 & 0 & 0 \end{pmatrix}$}
    \\
    \hline \hline \modelrow
    {} {$\upO(N_D)^2 \rtimes \upO(2) $} {2}
    {$\begin{matrix}
      - \frac i 2 g_1 \psi_a^{\dagger L}
      \left(\phi_1+i\phi_2\right)^{\vphantom{|}}
      \psi_a^{L}+\text{h.c.}
      \\
      - \frac i 2 g_2 \psi_a^{\dagger R}
      \left(\phi_1+i\phi_2\right)_{\vphantom{|}}
      \psi_a^{R}+\text{h.c.}
    \end{matrix}$}
    {$2_{1/2} \oplus 2_{1/2}$} {$2^\R \oplus 2'^{\;\R}$}
    {\scriptsize \setlength\arraycolsep{0.2pt} $\begin{pmatrix} g_1 & 0 & 0 & 0 \\ 0 & -g_1 & 0 & 0 \\ 0 & 0 & g_2 & 0 \\ 0 & 0 & 0 & -g_2 \end{pmatrix},
   \begin{pmatrix} 0 & g_1 & 0 & 0\\ g_1 & 0 & 0 & 0 \\ 0 & 0 & 0 & g_2 \\ 0 & 0 & g_2 & 0 \end{pmatrix}$}
    \\
    \hline \hline \modelrow
    {} {$\upO(N_D) \times \upO(2) $} {2}
    {$\begin{matrix}
      - \frac i 2 \frac{g_1+g_2}{2} \bar \psi_a^L
      \left(\phi_1+i\phi_2\right)^{\vphantom{|}}
      \psi_a^R+\text{h.c.}
      \\
      - \frac i 2 \frac{g_1-g_2}{2} \psi_a^{\dagger L}
      \left(\phi_1+i\phi_2\right)^{\vphantom{|}}
      \psi_a^R+\text{h.c.}
    \end{matrix}$}
    {$1 \oplus 1 \oplus 2_1$} {$4^\R$}
    {\scriptsize \setlength\arraycolsep{0.2pt} $\begin{pmatrix} 0 & 0 & g_1 & 0 \\ 0 & 0 & 0 & g_2 \\ g_1 & 0 & 0 & 0 \\ 0 & g_2 & 0 & 0 \end{pmatrix},
   \begin{pmatrix} 0 & 0 & 0 & -g_1 \\ 0 & 0 & g_2 & 0 \\ 0 & g_2 & 0 & 0 \\ -g_1 & 0 & 0 & 0 \end{pmatrix}$}
  \end{modeltab}
  \caption{An exhaustive list of $M=2$ GNY-like models that preserve $\upO(2)$ symmetry, along with their full symmetry groups.}
  \label{tab:gnylist2}
\end{table}

\begin{table}
  \small \centering
  \begin{modeltab}
    \hline \hline \modelrow
    {Chiral Heisenberg} {$\U(N_D) \times \upO(3)$} {1}
    {$\begin{matrix}
      - \frac i 2 g \begin{pmatrix} \bar\psi^L_a & \bar\psi^R_a \end{pmatrix}
      (\phi \cdot \sigma) \begin{pmatrix} \psi^L_a \\ \psi^R_a\end{pmatrix} = \\
      - \frac i 2 g \bar \Psi_a (i \tilde\gamma^3 \phi_1 + i \tilde\gamma^5  \phi_2 + \phi_3) \Psi_a
    \end{matrix}$}
    {4} {$4^\C$}
    {\scriptsize
     $\begin{pmatrix} 0 & 0 & g & 0 \\ 0 & 0 & 0 & g\\ g & 0 & 0 & 0 \\ 0 & g & 0 & 0 \end{pmatrix}$,
     $\begin{pmatrix} 0 & 0 & 0 & g \\ 0 & 0 & -g & 0\\ 0 & -g & 0 & 0 \\ g & 0 & 0 & 0 \end{pmatrix}$,
     $\begin{pmatrix} g & 0 & 0 & 0 \\ 0 & g & 0 & 0 \\ 0 & 0 & -g & 0 \\ 0 & 0 & 0 &-g \end{pmatrix}$}
   \\
   \hline \hline \modelrow
   {``\NewHeisenberg'' \\(new fixed point)} {$\upO(N_D) \times \upO(3)$} {1} {$ - \frac i 2 g \chi_a \left( \phi_1 \psi_a^1 + \phi_2 \psi_a^2 + \phi_3 \psi_a^3 \right)$}
   {$1 \oplus 3$} {$4^\R$}
   {\scriptsize
  $\begin{pmatrix} 0 & g & 0 & 0 \\ g & 0 & 0 & 0 \\ 0 & 0 & 0 & 0 \\ 0 & 0 & 0 & 0 \end{pmatrix}$,
$\begin{pmatrix} 0 & 0 & g & 0 \\ 0 & 0 & 0 & 0 \\ g & 0 & 0 & 0 \\ 0 & 0 & 0 & 0 \end{pmatrix}$,
$\begin{pmatrix} 0 & 0 & 0 & g \\ 0 & 0 & 0 & 0 \\ 0 & 0 & 0 & 0 \\ g & 0 & 0 & 0 \end{pmatrix}$}
  \end{modeltab}
  \caption{An exhaustive list of $M=3$ GNY-like models that preserve $\upO(3)$ symmetry, along with their full symmetry groups.  For these models, the $\U(N_D)$ and $\upO(N_D)$ subgroups do not experience symmetry enhancement in (2+1)d, so the semidirect product reduces to a direct product.}
  \label{tab:gnylist3}
\end{table}

\section{Perturbative fixed points} \label{sec:FPs}

In this section, we compute $\beta$-functions for the Yukawa and quartic couplings in the $4-\eps$ expansion to two-loop order, using the results from \cite{Machacek:1983tz,Machacek:1983fi,Machacek:1984zw}. For each model, we identify perturbative fixed points (simultaneous roots of the $\beta$-functions for small $\epsilon$) and compute their anomalous dimensions via the formulas in \cite{Machacek:1983tz} and \cite{Jack:1990eb}. Unitarity requires $\lambda$ to be positive and all $g$ to be real, so we restrict our attention to fixed points where this is true at leading order in $\eps$, setting aside nonperturbative solutions for future analysis. We also ignore here fixed points for which one or more of the $g$ couplings are zero, as these invariably correspond to models with a smaller number of interacting fermions.

Note that at higher orders in the $\epsilon$ expansion one would need to be careful to systematically incorporate counterterms that violate (3+1)d Lorentz symmetry. This is sketched in appendix A of~\cite{Jack:2024sjr}, but is not needed for the two-loop analysis that we perform here. 

\subsection{GNY, quarter GNY, and chiral Ising} \label{sec:betaM1}

For the most general $M=1$ model, where $g_1$, $g_2$, $g_3$, and $g_4$ are allowed to be distinct, the $\beta$-functions are 
\begin{align}
  \label{eq:asdf1} \begin{split}
    \beta_{g_i}
    &=
      -\frac{1}{2} g_i \epsilon
      + \frac{6 g_i^3 + N_D g_i\sum_j g_j^2}{32 \pi ^2}\\
    &\quad +\frac{
      \lambda ^2 g_i - 24 \lambda g_i^3 - 27 g_i^5
      - 21 N_D g_i^3\sum_j g_j^2 - 15 N_D g_i\sum_j g_j^4
      }{3072 \pi ^4}
  \end{split}
\end{align}
and
\begin{align}
  \label{eq:asdf2} \begin{split}
    \beta_{\lambda}
    &= - \lambda \epsilon
      + \frac{2 N_D\sum_j g_j^2 \left(\lambda -6\sum_j g_j^2\right)+3 \lambda ^2}{16 \pi ^2}\\
    &\quad+\frac{-9 \lambda ^2 N_D\sum_j g_j^2+21 \lambda  N_D\sum_j g_j^4+288 N_D\sum_j g_j^6-17 \lambda ^3}{768 \pi ^4}.
  \end{split}
\end{align}
This model has perturbative fixed points when
\begin{equation}
    g_1^2 = g_2^2 = g_3^2 = g_4^2 = \frac{16 \pi ^2}{4 N_D+6}\epsilon +  O(\eps^2)
\end{equation}
and
\begin{equation}
    \lambda^2 =\frac{8 \pi ^2 \left(\sqrt{16 N_D^2+528 N_D+36}-4 N_D+6\right)}{3 \left(4 N_D+6\right)}\epsilon + O(\epsilon^2),
\end{equation}
but we are free to choose the signs of each $g_i$. Up to isomorphism, this gives us three universality classes.

The first two of these are well-known: When all $g$s have the same sign, we obtain the ordinary GNY model (equation \ref{eq:gnyMajorana}), which has $\upO(4N_D)$ symmetry. And when  $g_1 = g_2 = -g_3 = -g_4$, we obtain the chiral Ising model (equation \ref{eq:yukCI}), which has $\upO(2N_D)^2 \rtimes \Z_2$ symmetry.

For $g_1 = g_2 = g_3 = -g_4$, we obtain a fixed point with $\upO(3N_D) \times \upO(N_D)$ symmetry and the Lagrangian
\begin{equation}
  \calL =   \calL_{\mathrm{free}} - \frac \lambda 4 \phi^4
  - i \frac g 2 \phi \left(
    \psi_a^1 \psi_a^1 + \psi_a^2 \psi_a^2 + \psi_a^3 \psi_a^3 - \psi_a^4 \psi_a^4,
    \right)
\end{equation}
which we shall refer to here as the \textbf{quarter GNY model}. It has been mentioned in two other recent classification papers (see table 4 of~\cite{Pannell:2023tzc} and section 12 of~\cite{Jack:2024sjr}), but has not been well-characterized beyond this. 

Aside from the different symmetry group, it displays very similar dynamics to the chiral Ising and GNY CFTs. At two-loop order, the scaling dimensions of $\psi$, $\sigma \sim \phi$, and $\varepsilon \sim \phi^2$ in all three of these models are identical:
\begin{align}
  \label{eq:dPsiQuarter}
  \Delta_\psi &=
  \frac{3}{2}-\frac{\left(4 N_D+5\right)}{2 \left(4 N_D+6\right)}\eps+
                \frac{N_D \left(3-328 N_D\right)+180+\left(2 N_D+33\right)\sqrt{4 N_D \left(N_D+33\right)+9}}{108 \left(4 N_D+6\right)^3} \eps^2,\\
  \label{eq:dSigQuarter}
  \Delta_\sigma
  &=1
    -\frac{3}{4 N_D+6} \eps
-\frac{208 N_D^2-57 N_D+\sqrt{4 N_D^2+132 N_D+9} \left(22 N_D+3\right)+9}{72 \left(2 N_D+3\right)^3}\epsilon ^2,
  \\
  \label{eq:dEpsQuarter}
  \Delta_\varepsilon &= 2+\frac{\epsilon  \left(\sqrt{4 N_D^2+132 N_D+9}-2 N_D-15\right)}{6 \left(2 N_D+3\right)}
\end{align}
(in agreement with \cite{fei2016yukawa} and table 1 of \cite{Mitchell:2024hix} when $N=4N_D$). This is what we expect---the sign should only affect fermion loops with an odd number of edges, which only appear in two-point functions at high loop order.

It is worth noting that, because $\beta_{g_i}$ vanishes when $g_i=0$, RG flows cannot change the sign of any of the Yukawa couplings. Therefore, neither the GNY nor chiral Ising nor quarter GNY models can be related by a perturbative RG flow.

\subsection{Chiral XY}

Table \ref{tab:gnylist1} contains three rows: the first is the chiral XY model (also called the Nambu--Jona--Lasinio--Yukawa model), and the second and third are new theories with two coupling constants each. As in the $M=1$ case, these new models only display perturbative fixed points when $g_1^2 = g_2^2$. But when this condition holds, both models reduce to the chiral XY model, albeit rewritten in a different basis (and unlike in the $M=1$ case, the signs of the couplings do not matter: we can always find a change of basis that flips the sign of $g_1$ or $g_2$). 

The chiral XY CFT is therefore the only GNY-like fixed point with two scalar fields. We direct the reader to \cite{Rosenstein:1993zf}, \cite{fei2016yukawa}, and \cite{Zerf:2017zqi} for calculations of its critical exponents and scaling dimensions. It is worth clarifying that in (2+1)d, this model displays a symmetry enhancement from $\U(N_D) \times \upO(2)$ to  $\upO(2N_D) \rtimes \upO(2)$. 

\subsection{Chiral Heisenberg and a new $M=3$ CFT}

\begin{figure}[th]
  \centering
  \includegraphics[width=\textwidth]{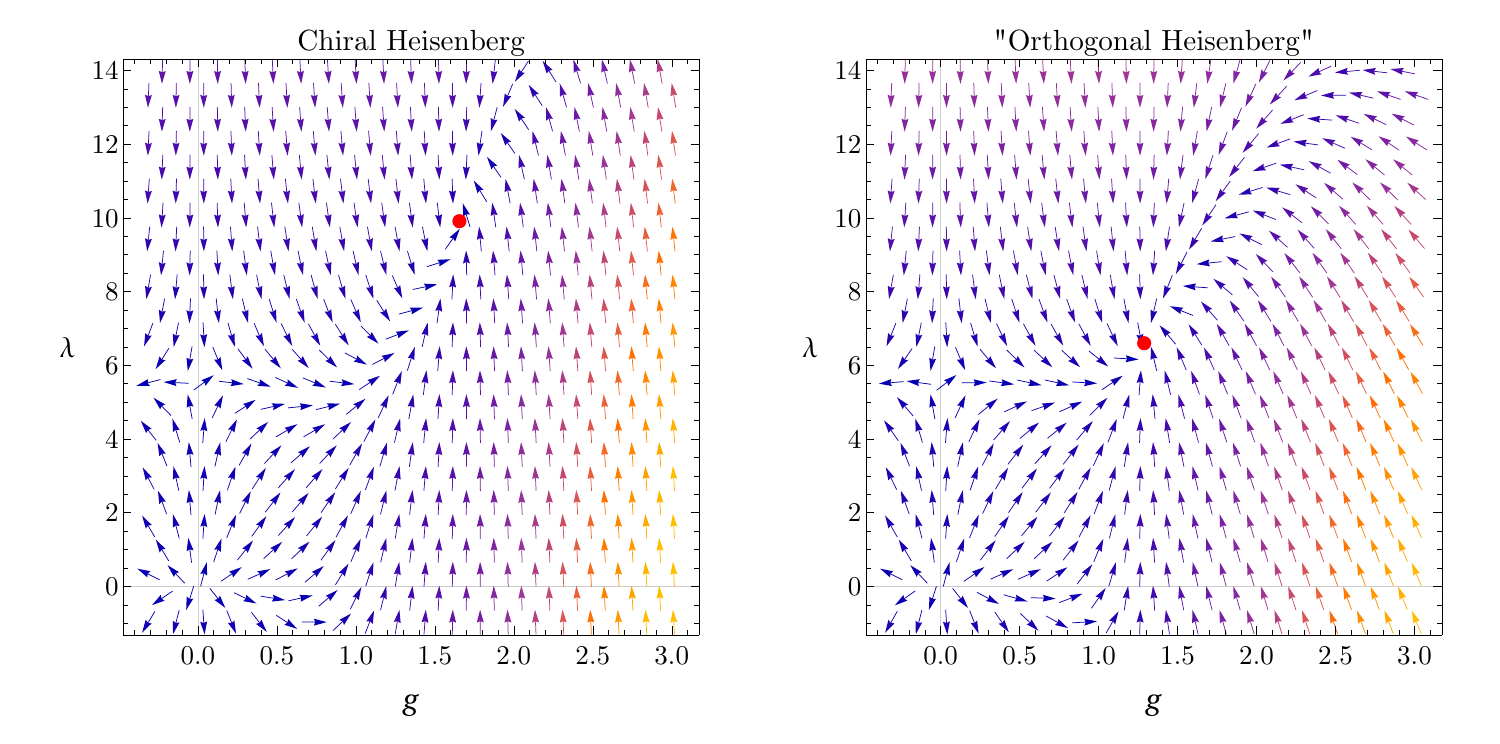}
  \caption{$\beta$-functions for the chiral Heisenberg and {\newHeisenberg} universality classes at $N_D=1$ and $\eps=0.1$. Perturbative fixed points with $g \ne 0$ are indicated in red (note that the free theory and Wilson--Fisher fixed points are also visible).}
  \label{fig:heisenbergBeta}
\end{figure}

There are two $M=3$ models, both of which have perturbative fixed points. The first is known as the chiral Heisenberg model \cite{Rosenstein:1993zf}, and its universailty class appears in the transition between semimetallic and antiferromagnetic phases in graphene \cite{janssen2014antiferromagnetic,Zerf:2017zqi}. The second, which we call the \textbf{{\newHeisenberg} model}, appears to be new. Its Lagrangian, corresponding to the $R$-matrices in the final row of table \ref{tab:gnylist1}, can be written as
\begin{equation}
  \calL =   \calL_{\mathrm{free}} - \frac \lambda 4 \phi^4
  - i \frac g 2 \chi_a \left(
    \phi_1 \psi_a^1 + \phi_2 \psi_a^2 + \phi_3 \psi_a^3
    \right).
\end{equation}
This model possesses an $\upO(N_D) \times \upO(3)$ symmetry (in contrast to the chiral Heisenberg model's $\U(N_D) \times \upO(3)$). Here, $3N_D$ of the Majorana spinors transform as a vector in the $\upO(3)$ subgroup.

The {\newHeisenberg} model's two-loop $\beta$-functions are
\begin{align}
  \beta_g &= -\frac{g}{2} \epsilon +\frac{(N_D+4)g^3}{16 \pi ^2} + \frac{g\lambda ^2-8 g^3 \lambda - (60 N_D+19)g^5}{1024 \pi ^4},\\
  \beta_\lambda &= -\lambda \eps
                  +\frac{4N_D g^2 \lambda- 24 N_D g^4 + 3 \lambda ^2}{16 \pi ^2}
                  +\frac{-6 N_D g^2 \lambda ^2 +8 N_D g^4 \lambda + 240 N_D g^6 -5 \lambda ^3}{256 \pi ^4}
\end{align}
and its scaling dimensions at order $\epsilon^2$ are
\begin{align}
  \label{eq:dChi31}
  \Delta_\chi
  &=\frac 3 2 -\frac{2 N_D+5}{4 \left(N_D+4\right)} \eps
    +\frac{-56 N_D^2+235 N_D+199 +\left(2 N_D+16\right)\sqrt{N_D^2+64 N_D+16} }{96 (N_D + 4)^3} \eps^2,\\ 
  \label{eq:dPsi31}
  \Delta_\psi
  &= \frac 3 2 -\frac{2 N_D+7}{4 \left(N_D+4\right)} \eps
    + \frac{-56 N_D^2+217 N_D+127+(2 N_D+16) \sqrt{N_D^2+64 N_D+16}}{288 (N_D + 4)^3} \eps^2,\\
 \label{eq:dSig31}
  \Delta_\sigma
  &=1
    +\frac{2}{N_D+4} \eps
    + \frac{380 N_D^2-85 N_D+128+\left(16 N_D+32\right)\sqrt{N_D^2+64 N_D+16}}{144 \left(N_D+4\right)^3} \eps^2,
\end{align}
and
\begin{align}
  \label{eq:dEps31}
    \Delta_\epsilon
    &= 2 + \frac{\sqrt{N_D^2+64 N_D+16}-N_D-20}{6 \left(N_D+4\right)} \eps \nonumber\\
    &\qquad+ \frac{{\scriptstyle-48 N_D^3+752 N_D^2+1142 N_D+1664+}\frac{48 N_D^4+3052 N_D^3+24987 N_D^2+6432 N_D+6656}{\sqrt{N_D^2+64 N_D+16}}}{216 \left(N_D+4\right)^3} \eps^2. 
\end{align}
As far as the $4-\eps$ expansion is concerned, this fixed point is qualitatively similar to that of the chiral Heisenberg model, but possesses slightly different conformal data and occurs at a different value of the couplings. For a comparison, see figure   \ref{fig:heisenbergBeta}, as well as the anomalous dimensions in \cite{Rosenstein:1993zf} and \cite{Zerf:2017zqi}.

\section{Spontaneous symmetry breaking and CPT}

So far, we have only classified these models by their global symmetries, ignoring discrete spacetime symmetries. We will devote this subsection to the latter.

First, we need to clarify what we mean by parity, and how we distinguish it from other discrete symmetries. A parity transformation must correspond to reflection along a certain vector---say, $\hat x$---which is implemented by a unitary operator $P$. Rotations in planes containing $\hat x$ must anticommute with $P$,  and all other spacetime symmetries must commute with it. In other words,
\begin{equation}
  \label{eq:pardef}
  P M^{\mu\nu} P^{-1} = 
  \begin{cases}
    -M^{\mu\nu} &\text{if $\mu = x$ or $\nu = x$}\\
    +M^{\mu\nu} &\text{otherwise.}
  \end{cases}
\end{equation}

But if our theory also has a $\Z_2$ global symmetry (call the nontrivial element $Z$), there is an ambiguity. Internal symmetries must commute with spacetime symmetries, so both $P$ and $PZ$ satisfy the condition in equation \ref{eq:pardef}. We could take either to be our ``parity.''

This is important for determining which symmetries are broken in a phase transition. With regard to selection rules, theories that break $Z$ and $P$ but preserve $PZ$ are equivalent to theories that break $Z$ and $PZ$ but preserve $P$. Observers within a theory cannot see the Lagrangian, so for a theory to be ``parity-violating'' in the traditional sense, there must be \emph{no} selection rules that are enforced by a unitary operator satisfying (\ref{eq:pardef}). In other words, \textbf{parity is broken when there are no remaining discrete symmetries that could be called ``parity.''} % todo: reword

We see this ambiguity play out in the chiral Ising model. When $\phi$ acquires a VEV, the Yukawa coupling in equation \ref{eq:yukCI} becomes a fermion mass term of the form
\begin{equation}
  - \frac {i g} 2 \langle \phi \rangle (\psi^L \psi^L  - \psi^{R} \psi^{R}).
\end{equation}
There is a ``na\"ive'' definition of parity, which we'll call $P_N$, under which $\phi$, and $\psi^L\psi^L$, and $\psi^R\psi^R$ are all odd. So this theory would appear to break parity, as is the case for the standard GNY model. But the fermion bilinears are interchanged under the discrete part of the global symmetry, so $\phi$ and $\psi^L\psi^L - \psi^R\psi^R$ are both odd under $Z$ and even under $Z P_N$. The former is broken and the latter is preserved, so critical points in the chiral Ising universality class break $\Z_2$ global symmetry while preserving parity (which we take to be $P = Z P_N$). This is consistent with the standard definition of parity for \emph{Dirac} spinors.

Charge conjugation and time reversal are much simpler. Majorana fermion bilinears are manifestly invariant under $C$ symmetry, so none of the models we study spontaneously break $C$. And as for $T$ symmetry, we can simply use $T = CP$.

\begin{table}
  \small \centering
  \begin{tabular}{|c|c|c|c|}
    \hline \bigstrut
    Critical point & Global symmetry & Breaking when $\langle \phi \rangle \ne 0$  & Scaling dimensions \\
    \hline
    \hline \bigstrut
    GNY & $\upO(4N_D)$ & parity, time-reversal & \cite{fei2016yukawa}, \cite{Erramilli:2022kgp}, \cite{Mitchell:2024hix}\\
    \hline \bigstrut
    Quarter GNY & $\upO(3N_D) \times \upO(N_D)$ & parity, time-reversal & (\ref{eq:dPsiQuarter}) - (\ref{eq:dEpsQuarter})\\
    \hline \bigstrut
    Chiral Ising & $\upO(2N_D)^2 \rtimes \Z_2 $ & $\Z_2$ & \cite{fei2016yukawa}, \cite{Zerf:2017zqi}  \\
    \hline \bigstrut
    Chiral XY & $\upO(2N_D) \rtimes \upO(2)$ & $\upO(2) \to \Z_2$ & \cite{fei2016yukawa}, \cite{Zerf:2017zqi} \\
    \hline \bigstrut
    Chiral Heisenberg & $\U(N_D) \times \upO(3)$ & $\upO(3) \to \upO(2)$ & \cite{Zerf:2017zqi}\\
    \hline \bigstrut
    \NewHeisenberg & $\upO(N_D) \times \upO(3)$ & $\upO(3) \to \upO(2)$ & (\ref{eq:dChi31}) - (\ref{eq:dChi31})\\
    \hline
  \end{tabular}
  \caption{All critical points, along with the symmetries broken by the order parameter.}
  \label{tab:ssb}
\end{table}

For all CFTs studied in this paper, the transition spontaneously breaks parity when the spectra of $R$ are asymmetric about zero---otherwise, a rotation on $\psi$ can be used to reverse the sign of the fermion mass term. This is the case for the GNY and quarter GNY models, for which $\phi$ is a parity-odd flavor singlet. For all other models, $\phi$ is parity-even but charged under the global symmetry group.  A list of all models, along with their symmetry breaking patterns, can be found in table \ref{tab:ssb}.

In the GNY, quarter GNY, and {\newHeisenberg} models, the Yukawa couplings are not Lorentz-invariant in (3+1)d. They are, however, Lorentz-invariant in (2+1)d, and could in principle arise as quantum critical points of low-dimensional materials.

\section{Discussion and future directions}

We have identified six distinct universality classes that can be obtained from Dirac fermions with a global symmetry and a scalar order parameter. Four of them---the GNY, chiral Ising, chiral XY, and chiral Heisenberg models---are well-known in the literature. The other two, which we dub the ``quarter GNY'' and ``\newHeisenberg'' models, are not well studied, and it remains to be seen whether they can be easily realized in a (2+1)d condensed matter system.

Of these six models, only the GNY model has been fully bootstrapped \cite{Erramilli:2022kgp,Mitchell:2024hix}. In the future, we hope to obtain islands in scaling dimension and OPE space for the other five. This task is feasible with current bootstrap technology, and now that the representation theory of these models is understood, we can extend the bootstrap code from \cite{Erramilli:2022kgp} and \cite{Mitchell:2024hix} (note that the spacetime tensor structures are identical, and only the flavor structures need to be adapted).

It is also an open question whether these theories have other fixed points which are unstable or nonunitary in the $4-\eps$ expansion, but correspond to sensible theories in 3 dimensions. The GNY model, for example, has a perturbative fixed point with $\lambda \le 0$ at 2-loop order \cite{fei2016yukawa}. This could be investigated using bootstrap methods by constructing a navigator function \cite{Reehorst:2021ykw,Liu:2023elz}, sampling it over a wide range of scaling dimension space, and examining which local minima correspond to fixed points with $\lambda \le 0$ or $g^2 \le 0$.

In the future, we plan to extend this classification to models which explicitly break the $\upO(M)$ symmetry to some subgroup,  models in which the $\upO(M)$ symmetry is gauged, and models in which the fermions do not transform in the fundamental representation of $\U(N_D)$ (see the $M=5$ case in the appendix). One can also extend this analysis to models which involve a more general number of (2+1)d Majorana fermions. It would also be interesting to incorporate higher-loop results for the GNY $\beta$-functions and study the structure of RG flows connecting the various models. 

\acknowledgments

We thank Li-Yuan Chiang, Rajeev Erramilli, Mark Gonzalez, Yin-Chen He, Luca Iliesiu, Petr Kravchuk, Aike Liu, Ian Moult, Gordon Rogelberg, Slava Rychkov, David Simmons-Duffin, Robin Tsai, and Naveen Umasankar for discussions. The authors were supported by Simons Foundation grant 488651 (Simons Collaboration on the Nonperturbative Bootstrap) and DOE grant DE-SC0017660.

\appendix

\section{Real representation theory of $\Spin(M)$}\label{app:reps}

\subsection*{$M=1$}

The only irrep is the trivial rep, so $\rho_\Psi = 1 \oplus 1 \oplus 1 \oplus 1$, and we can tune up to 4 distinct couplings.

\subsection*{$M=2$}

We take $\Spin(2)$ is a copy of $\SO(2)$ that is ``twice as large''---in other words, it has the same irreps, but the frequency is now permitted to be a half-integer.\footnote{This definition is worth clarifying because for $M=2$ because, while $\Spin(2)$ is a double cover of $\SO(2)$, it is not the universal cover.} Over $\R$, these are the trivial irrep $1$ and the nontrivial irreps
\begin{equation}
  2_n = \begin{pmatrix}
    \cos (n \theta) & -\sin (n \theta) \\
    \sin (n \theta) & \cos (n \theta)
  \end{pmatrix} \text{ for } n \in \Z^+/2
\end{equation}
which satisfy
\begin{equation}
  \begin{split}
    1 \otimes 1 &= 1,\\
    1 \otimes 2_n &= 2_n,\\
    2_n \otimes 2_n &= 1 \oplus 1 \oplus 2_n,\\
    2_m \otimes 2_n &= 2_{|m-n|} \oplus 2_{m+n}.
  \end{split}
\end{equation}

The vector representation is $2_1$, so the possible decompositions of $\rho_\Psi$ (which do not lead to a vanishing $R$) are:
\begin{equation}
  \begin{split}
    1 \oplus 1 \oplus 2_{1/2} & \rightarrow 1 \text{ coupling (but singular)}, \\
    1 \oplus 1 \oplus 2_1 & \rightarrow 2 \text{ couplings},\\
    2_{1/2} \oplus 2_{1/2} & \rightarrow 2 \text{ couplings},\\
    2_n \oplus 2_{n+1} & \rightarrow 1 \text{ coupling}.
  \end{split}
\end{equation}
For the representation marked ``singular,'' only the lower-right 2x2 block of the $R$ matrices can be nonzero, so some of the fermions will decouple from the rest of the theory.

\subsection*{$M=3$}

The represention theory of  $\Spin(3) = \SU(2)$ over $\C$ is standard material. The irreps over $\R$ are in 1-1 one correspondence with the irreps over $\C$, as follows:
\begin{itemize}
\item If $\rho_\C$ can be made real, it is already an irrep over $\R$, so  $\rho_\R = \rho_\C$ .
\item If $\rho_\C$ is complex and not self-conjugate, $\rho_\R = \rho_\C \oplus \bar{\rho_\C}$ .
\item If $\rho_\C$ is complex and self-conjugate (quaternionic), $\rho_\R = \rho_\C \oplus \rho_\C$ .
\end{itemize}
In $\SU(2)$, all odd-dimensional (i.e.~bosonic) irreps are real, and all even-dimensional (i.e.~fermionic) irreps are quaternionic. So the real irreps are $1$, $3$, $4 = 2_\C \oplus 2_\C$, $5$, and so on. The possible $\rho_\Psi$ are therefore:
\begin{equation}
  \begin{split}
     1 \oplus 3 & \rightarrow 1 \text{ coupling (} 3 \in 3 \otimes 3 \text{ is antisymmetric)},\\
    4 & \rightarrow 1 \text{ coupling}.
  \end{split}
\end{equation}

The first corresponds to our {\newHeisenberg} model, the second to the chiral Heisenberg model.

\subsection*{$M=4$}

Recall that $\Spin(M) = \SU(2) \times \SU(2)$. It has a real, 4d vector irrep $4 = 2_L \otimes 2_R$, but this irrep does \textbf{not} appear in the product of any 4d representations over $\R$. Thus, there are no models.

\subsection*{$M=5$}

The smallest irrep of $\Spin(5) = \Sp(4) = B_2$ is $4$, which is self-conjugate. The Frobenius--Schur indicator (computed in both LieART \cite{Feger:2019tvk} and  LiE \cite{van1994lie}) is $-1$, so the irrep is quaternionic, and there are no GNY-like models.

There \emph{is} a putative $\upO(5)$ model in which the fermions transform in the antisymmetric $\bf 10$ representation, given by:
\begin{equation}
  \calL \supset \psi_{ab} (\epsilon^{abcde} \phi_c) \psi_{de},
\end{equation}
but it cannot be obtained by starting with Dirac spinors in (3+1)d and contracting $\phi$ with a $4\times 4$ spinor matrix. And regardless, it lies outside the scope of this classification because the fermions transform in the adjoint, rather than the fundamental. In the future, we hope to extend our classification to include this model and those like it.

\subsection*{$M \ge 6$}

$\Spin(6)$ has a 4d irrep over $\C$, but not over $\R$. For $M \ge 6$ , the smallest nontrivial irreps have dimension greater than 4, so there are no GNY-like models.

\bibliographystyle{JHEP}
\bibliography{refs}

\end{document}